\documentclass[conference]{IEEEtran}
\usepackage{cite}

\ifCLASSINFOpdf
\else
   \usepackage[dvips]{graphicx}
   \graphicspath{{./Graphics/}}
   \DeclareGraphicsExtensions{.eps}
\fi
\usepackage{amsmath,amsfonts,amsthm} 
\newtheorem{lemma}{Lemma}

\newtheorem{proposition}{Proposition}
\newtheorem{definition}{Definition}
\ifCLASSOPTIONcompsoc
  \usepackage[caption=false,font=normalsize,labelfont=sf,textfont=sf]{subfig}
\else
  \usepackage[caption=false,font=footnotesize]{subfig}
\fi

\DeclareMathOperator*{\Var}{Var}

\begin{document}
%
\title{Age of Information: The Gamma Awakening}

\author{\IEEEauthorblockN{Elie Najm and Rajai Nasser}
\IEEEauthorblockA{LTHI, EPFL, Lausanne, Switzerland\\
Email: \{elie.najm, rajai.nasser\}@epfl.ch}
}


%


\maketitle

\begin{abstract}
Status update systems is an emerging field of study that is gaining interest in the information theory community. We consider a scenario where a monitor is interested in being up to date with respect to the status of some system which is not directly accessible to this monitor. However, we assume a source node has access to the status and can send status updates as packets to the monitor through a communication system. We also assume that the status updates are generated randomly as a Poisson process. The source node can manage the packet transmission to minimize the age of information at the destination node, which is defined as the time elapsed since the last successfully transmitted update was generated at the source. We use queuing theory to model the source-destination link and we assume that the time to successfully transmit a packet is a gamma distributed service time. We consider two packet management schemes: LCFS with preemption and LCFS without preemption. We compute and analyze the average age and the average peak age of information under these assumptions. Moreover, we extend these results to the case where the service time is deterministic.
\end{abstract}


%
\IEEEpeerreviewmaketitle

\section{Introduction}
\label{Intro}
In status update systems, one or several sources send information updates to one or several monitors at a certain effective rate $\lambda$. Naturally, the goal of this process is to ensure that the status updates are as timely as possible at the receiver side. For this purpose, \cite{realtime} uses the term \emph{age}, to refer to the time elapsed since the generation --- at instant \emph{u(t)} --- of the newest packet available at the receiver. Formally, the age of such packet is $\Delta(t)=t-u(t)$ and the timeliness requirement at the monitor corresponds to a small average age. Indeed, real-time status updating can be modeled as a source feeding packets at rate $\lambda$ to a queue which delivers them to the monitor with some delay. Hence, the requirement at the destination translates into finding the optimal transmission scheme and/or the optimal effective update rate $\lambda$ at the source that minimizes
\begin{equation}
 \label{eq1}
 \Delta = \lim_{\tau\to\infty}\frac{1}{\tau}\int_0^\tau \Delta(t) \mathrm{d}t.
\end{equation}
However, numerous factors affect the evaluation of ~\eqref{eq1} such as the model of the source update process, the number of sources, the model of the queue, the number of queues available, etc.

Kaul et al. in \cite{realtime} solve one aspect of the problem where they consider a single source generating packets as a rate $\lambda$ Poisson process feeding them to a single First Come First Served (FCFS) queue with exponential service time. Using Kendall notation, this is an FCFS $M/M/1$ system. Moreover, the authors also consider the cases of deterministic source and exponential service time, i.e., FCFS $D/M/1$ system, as well as random source and deterministic service time, i.e., FCFS $M/D/1$ system.

Yates and Kaul in \cite{multiple} generalize the problem solved in \cite{realtime} by considering the presence of multiple sources sending updates through the same FCFS queue to the same monitor. Along the same generalization direction as in \cite{multiple}, one may ask: what would happen if we increase the number of queues available, i.e., if the source is able to serve multiple updates at the same time? This question is tackled in \cite{random}, where a single Poisson process is sending updates over an infinite number of queues with exponential service time.\\
\indent However, in these aforementioned works, the authors mostly consider FCFS queues. One would focus on Last Come First Served (LCFS) type of queues since they are intuitively more suitable for the problem in hand: we are interested in delivering the newest update to the monitor, which means we gain more by sending the \lq\lq youngest\rq\rq\ packet in the queue first. This idea is developed in \cite{queues} where the authors derive an expression for \eqref{eq1} by treating the following two models while assuming exponential interarrival and service time: $(i)$ LCFS queue without preemption; if the queue is busy, any new update will have to wait in a buffer of size 1. This means that the new update will replace any older packet already waiting to be served. $(ii)$ LCFS with preemption, where contrary to the first case, any new update will prompt the source to drop the packet being served and start transmitting the newcomer. In ~\cite{queues}, it is shown that an LCFS queue with preemption achieves a lower average age compared to the model without preemption. However, both models outperform the FCFS model presented in \cite{realtime}.

In this paper, we also consider these last two schemes in order to derive closed form expressions for \eqref{eq1}. However, the main novelty is the assumption of a gamma distribution for the service time in age of information problems. The motivation for such a distribution is twofold: 
\begin{itemize}
 \item Based on the classical applications of gamma distributions in queuing theory, these distributions can be seen as a reasonable approximation if we want to model relay networks. Indeed, in such network, a transmitter and a receiver are separated by $k$ relays with each relay taking an exponential amount of time to complete transmission to the next hop. This means that the total transmission time is the sum of $k$ independent exponential random variables which induces a gamma distribution.
 \item As we will see later, a deterministic random variable can be seen as the limit of a sequence of gamma distributed random variables. Therefore, one can study the performance of the LCFS-based schemes under deterministic service time by taking the limit of the result obtained for a gamma distributed service time. Although this is an indirect method of calculating ~\eqref{eq1}, it is simpler than the direct approach.
\end{itemize}
This paper is organized as follows: in Section~\ref{PC} we present the preliminary results that will be used throughout our work and define the average peak age as an alternative metric. In Section~\ref{preempt} we derive the closed form expressions for both the average age and the average peak age when assuming an LCFS scheme with preemption. On the other hand, Section~\ref{nopreempt} computes the formulas for these quantities when considering an LCFS queue without preemption. In these last two sections the service time is assumed to be gamma distributed. However, in Section~\ref{deterministic} we calculate the two ages for a deterministic service time for each of the two schemes. Finally, Section~\ref{numerical} presents numerical simulations that validate our theoretical results.


\section{Preliminaries}
\label{PC}
\subsection{General definitions}

\begin{figure*}[!t]
\centering
\subfloat[Age of information for LCFS with preemption scheme]{\includegraphics[scale=0.5]{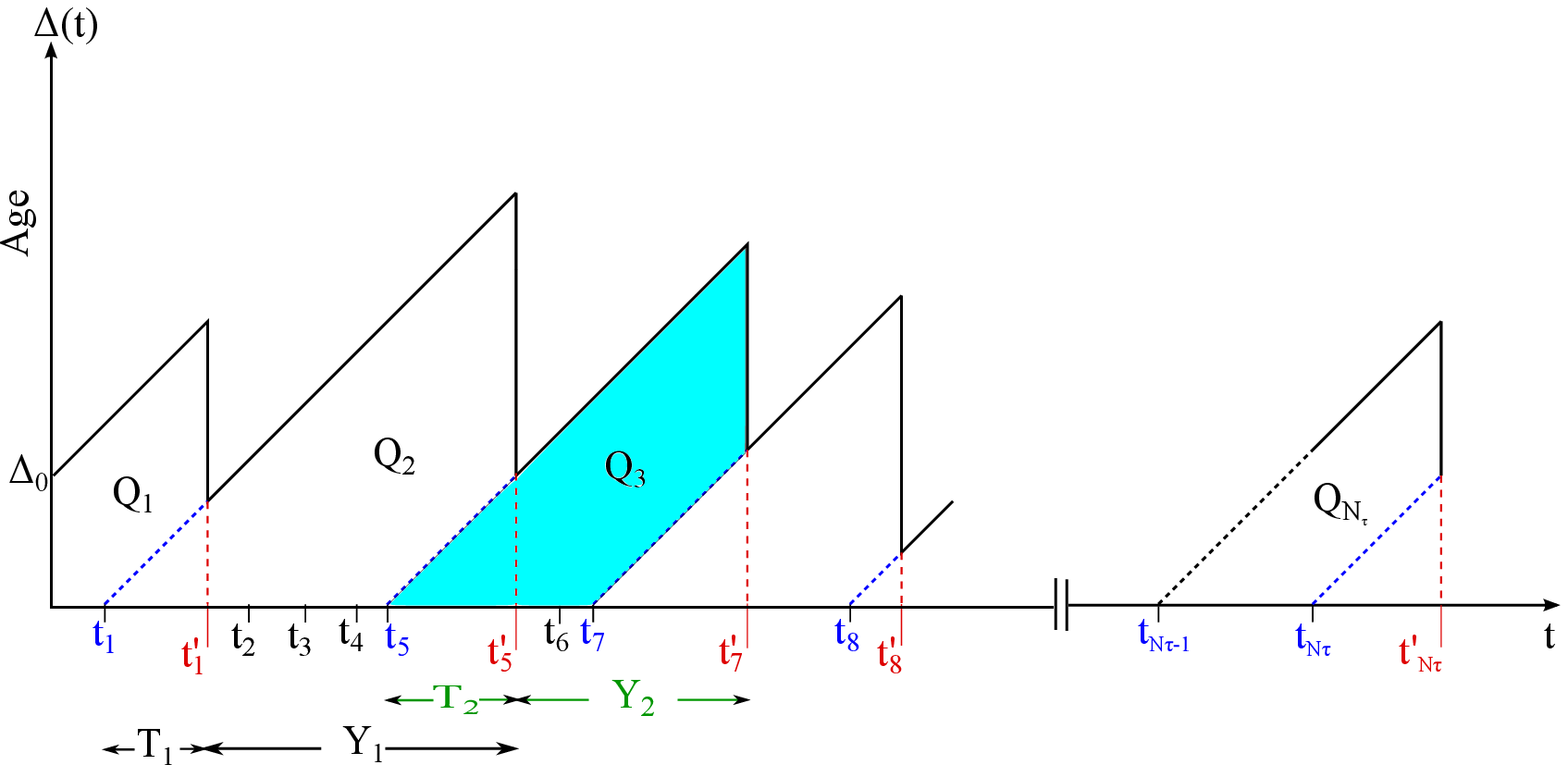}%
\label{fig1a}}
\hfil
\subfloat[Age of information for LCFS without preemption scheme]{\includegraphics[scale=0.5]{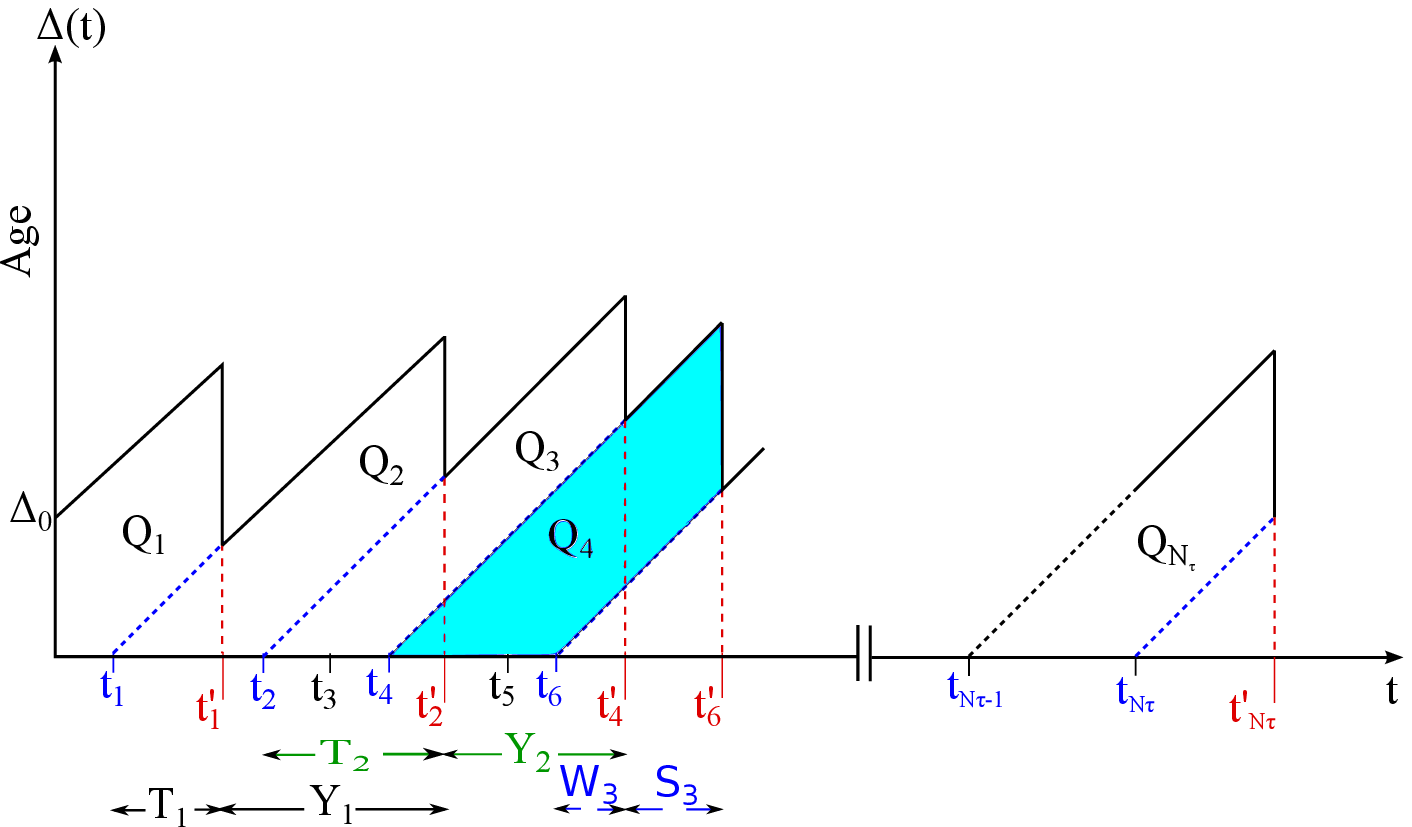}%
\label{fig1b}}
\caption{Variation of the instantaneous age for both schemes}
\label{fig1}
\end{figure*}

As we have seen, our two schemes of interest are LCFS with preemption and LCFS without preemption. The variation of the instantaneous age for these two scenarios is given in Figure~\ref{fig1}. The saw-tooth pattern depicted in those figures is due to the following behavior of the age.
Let $t_i$ be the time the $i^{th}$ packet is generated and let $t'_i$ be the time the $i^{th}$ packet is received (if it is successfully received). Moreover, without loss of generality, we assume the beginning of observation is at time $t= 0$ and that the queue is empty at this instant with an initial age of $\Delta_0$. The age $\Delta(t)$ increases linearly with time and is set to a smaller value when a packet is received. Hence, the instantaneous age is equal to the current time minus the generation time of the newest of the received packets.

It is important to note that in both schemes of interest, some packets might be dropped. Hence we call the packets that are not dropped, and thus delivered to the receiver, as \lq\lq successfully received packets\rq\rq\ or \lq\lq successful packets\rq\rq. In addition to that, we also define: $(i)$ $I_i$ to be the true index of the $i^{th}$ successfully received packet, $(ii)$ $Y_i= t'_{I_{i+1}}-t'_{I_i}$ to be the interdeparture time between two consecutive successfully received packets, $(iii)$ $X_i = t_{I_i+1} - t_{I_i}$ to be the interarrival time between the successfully transmitted packet and the next generated one (which may or may not be successfully transmitted), so $f_X(x) = \lambda e^{-\lambda x}$, $(iv)$ $T_i$ to be the system time, or the time spent by the $i^{th}$ successful packet in the queue and $(v)$ $N_\tau=\max\left\{n;t_{I_n}\leq\tau\right\}$, the number of successfully received packets in the interval $[0,\tau]$.

\subsection{Computing the Average Age}
Using these quantities and Figures~\ref{fig1a} and \ref{fig1b}, the authors in \cite{queues} show that 
\begin{align}
\label{eq2}
 \Delta &= \lim_{\tau\to\infty} \frac{1}{\tau}\int_0^{\tau}\Delta(t)dt=\lambda_e \mathbb{E}(Q_i),
\end{align}
where $\lambda_e$ is the effective update rate and $\mathbb{E}(Q_i)$ is the expected value of the area $Q_i$ at steady state. Hence, we need to determine these two quantities.
 
\subsubsection{Computing the effective rate}

As stated in \cite{ontheage},
\begin{equation}
 \label{eq4}
 \lambda_e = \lambda\cdot\mathbb{P}\left(\{\text{packet is received successfully}\}\right)
\end{equation}
where $\mathbb{P}\left(\{\text{packet is received successfully}\}\right)$ is the probability that a packet in the queue will be delivered to the receiver.

\subsubsection{Computing $\mathbb{E}(Q_i)$}
Based on Figures~\ref{fig1a} and ~\ref{fig1b}, it was shown in \cite{queues} that
\begin{align}
 \label{eq6}
 \mathbb{E}(Q_i) 
		 &= \mathbb{E}\left(T_{i-1}Y_{i-1}\right)+\mathbb{E}\left(\frac{Y_{i-1}^2}{2}\right).
\end{align}

\subsection{Computing the Average Peak Age}
Another metric of interest is the average peak age. We define the peak age as $$\displaystyle P_i=\lim_{\substack{t\to t'_{I_i}\\t<t'_{I_i}}} \Delta(t),$$ which is the value of the instantaneous age just before it is reduced by the reception of the $i^{th}$ successful packet.
From Figures~\ref{fig1a} and \ref{fig1b}, we can deduce that the peak age can be written as $P_i = T_{i-1} + Y_{i-1}$.
Therefore, the average peak age is given by:
\begin{equation}
 \label{eq9}
 \mathbb{E}(P_i) = \mathbb{E}(T_{i-1}) + \mathbb{E}(Y_{i-1}).
\end{equation}

\subsection{Defining the service time}

All the above results were obtained without any assumption on the service time. However, as we have discussed before, this paper studies two models for the service time: a gamma distributed service time with parameters $(k,\theta)$ and a deterministic service time. Here is a brief description of the gamma distribution.
\begin{definition}
 \label{def1}
 A random variable $S$ with gamma distribution $\Gamma(k,\theta)$ has the following probability density function: $$f_S(s)=\frac{s^{k-1}e^{-\frac{s}{\theta}}}{\theta^k\Gamma(k)}.$$ The Erlang distribution $E(k,\theta)$ is a special case of the gamma distribution where $k\in \mathbb{N}$.
\end{definition}
Such random variable has a mean of $\mathbb{E}(S)=k\theta$ and a variance $\Var(S)=k\theta^2$. These quantities will come in handy later on. Another important property of gamma random variables is given by the following lemma:
\begin{lemma}
 \label{lemma3}
 Suppose $S_n\sim\Gamma(k_n,\theta_n)$ is a sequence of random variables such that $\mathbb{E}(S_n)=\frac{1}{\mu}$, for some $\mu>0$. Then the sequence $S_n$ converges in distribution to a deterministic variable $Z$ as $k$ becomes very large, i.e, $$S_n \overset{d}{\to} Z,\ \text{as\ } k\to\infty,$$ where $Z=\frac{1}{\mu}$ with probability 1.
\end{lemma}
The above lemma obviously still holds if $S_n\sim E(k_n,\theta_n)$. This lemma provides an additional motivation for studying the average age and the average peak age under the assumption of a gamma distributed service time since we can easily extend the results to the deterministic service time model by letting $k\to\infty$.

\section{Age of information for LCFS with preemption}
\label{preempt}
In this section we will compute the average age $\Delta$ and the average peak age $\mathbb{E}(P_k)$ for the Last Come First Served (LCFS) scheme with preemption and a gamma distributed service time. As we have seen before, in this scenario any packet being served is preempted if a new packet arrives and the new packet is served instead. Hence, the number of packets in the queue can be modeled as a continuous-time two-state semi-Markov chain depicted in Figure~\ref{fig2}.

The 0-state corresponds to the state where the queue is empty and no packet is being served while the 1-state corresponds to the state where the queue is full and is serving one packet. However, given that the interarrival time between packets is exponentially distributed with rate $\lambda$ then one spends an exponential amount of time $X$ in the 0-state before jumping with probability 1 to the other state. Once in the 1-state, two independent clocks are started: the gamma distributed service time clock of the packet being served and the rate $\lambda$ memoryless clock of the interarrival time between the current packet and the next one to be generated. We jump back to the 0-state if the service time clock happens to tick before that of the interarrival time. Given that the interarrival times between packets are i.i.d as well as the service time of each packet, then the probability to jump from the 1-state to the 0-state does not depend on the index of the current packet. Hence, the jump from the 1-state to the 0-state occurs with probability $p = \mathbb{P}(S<X)$, where $S$ is a generic gamma distributed service time and $X$ is a generic rate $\lambda$ memoryless interarrival time which is independent of $S$. On the other hand, if the interarrival time clock happens to tick before the service time clock then the current packet being served is preempted and the new generated packet takes its place in the queue. Therefore, we stay in the 1-state and the two clocks are started anew independently from before. This explains the $1-p$ probability seen in Figure~\ref{fig2} for staying in the 1-state.

Given that the probability $p$ will be useful in the computation of the average age as well as the average peak age, we start by deriving its expression here:
\begin{align}
 \label{eq10}
 p = \mathbb{P}(S<X) &=\left(\frac{1}{1+\lambda\theta}\right)^k.
\end{align}
Now we are ready to derive the two age metrics.

\subsection{Average age}
We start by deriving the expression for the average age. We need to compute two quantities for this purpose: $\mathbb{E}(Q_i)$ and the effective rate $\lambda_e$. 

\subsubsection{Computing $\mathbb{E}(Q_i)$}
Using \eqref{eq6}, we obtain:
\begin{align}
\label{eq13}
 \mathbb{E}(Q_i) &= \mathbb{E}\left(T_{i-1}Y_{i-1}\right)+\mathbb{E}\left(\frac{Y_{i-1}^2}{2}\right)\nonumber\\
		 &= \mathbb{E}\left(T_{i-1}\right)\mathbb{E}\left(Y_{i-1}\right)+\mathbb{E}\left(\frac{Y_{i-1}^2}{2}\right).
\end{align}
The second equality comes from the fact that $T_i$ and $Y_i$ are independent (since the interarrival time is exponential and hence memoryless). In fact, the $i^{th}$ successful packet leaves the queue empty and hence $Y_i=\hat{X}_i+Z_i$ where $\hat{X}_i=X_i-T_i$ is the remaining of the interarrival time (between the departure of the $i^{th}$ successful packet and the arrival of the next generated one) and $Z_i$ is the time for a new packet to be successfully delivered. $Z_i$ does not overlap with $T_i$ and thus is independent from it. As for $\hat{X}_i$, we also get that it is independent of $T_i$. To prove this we notice that for a successfully received packet $i$ the joint distribution $f_{X_i,T_i}(x,t)$ can be written as
\begin{equation}
 \label{eq11}
 f_{X_i,T_i}(x,t) = \left\{\begin{array}{ll}
                       0 & \text{if\ }x<t\\
                       \frac{f_{X,S}(x,t)}{\mathbb{P}(S<X)} & \text{if\ }x>t
                      \end{array}\right.,
\end{equation}
where $X$ and $S$ are the generic independent interarrival time and service time respectively. Now, using a change of variable we get
\begin{align}
 \label{eq12}
 f_{\hat{X}_i,T_i}(\hat{x},t) &= f_{X_i-T_i,T_i}(\hat{x},t)= f_{X_i,T_i}(\hat{x}+t,t)\nonumber\\
			      &= \left\{\begin{array}{ll}
                       0 & \text{if\ }x<0\\
                       \frac{f_{X,S}(\hat{x}+t,t)}{\mathbb{P}(S<X)} & \text{if\ }x>0
                      \end{array}\right.\nonumber\\
			      &= \left\{\begin{array}{ll}
                       0 & \text{if\ }x<0\\
                       h(\hat{x})g(t) & \text{if\ }x>0
                      \end{array}\right..
\end{align}
\eqref{eq12} shows that $\hat{X}_i$ and $T_i$ are indeed independent. Moreover, one can show that $\hat{X_i}$ is exponential with rate $\lambda$. Given that $\hat{X}_i$ and $Z_i$ are both independent from $T_i$, then $Y_i$ and $T_i$ are also independent.

From now on we will drop the subscript index since at steady state $T_{i-1}$ and $T_i$ have same the distribution, which is also the case for $Y_{i-1}$ and $Y_i$.
The following lemma will be used to evaluate \eqref{eq13}:
\begin{lemma}
\label{lemma1}
 Let $G$ be gamma distributed with parameters ($k$,$\theta$) and $F$ be a rate $\lambda$ exponential random variable independent of $G$. Then, conditioned on the event $\{G<F\}$, the distribution of $G$ becomes gamma with parameters $\left(k,\frac{\theta}{1+\lambda\theta}\right)$.
 \begin{equation}
  \label{eq14}
  f_{G/G<F}(t) = \frac{t^{k-1}e^{-t\left(\frac{1+\lambda\theta}{\theta}\right)}}{\left(\frac{\theta}{1+\lambda\theta}\right)^k\Gamma(k)}.
 \end{equation}
\end{lemma}
\begin{proof}
 In order to prove this Lemma we will compute the probability density function $f_{G|G<F}$:
 \begin{align*}
  f_{G|G<F}(t) &= \lim_{\epsilon\to 0} \frac{\mathbb{P}(t\leq G<t+\epsilon|G<F)}{\epsilon}\\
	       &\overset{(a)}{=} \lim_{\epsilon\to 0} \frac{\mathbb{P}(t\leq G<t+\epsilon)\mathbb{P}(G<F|t<G<t+\epsilon)}{\epsilon\mathbb{P}(G<F)}\\
	       &= f_G(t)\frac{\mathbb{P}(F>t)}{\mathbb{P}(G<F)}\\
	       &\overset{(b)}{=} f_G(t)\frac{e^{-t\lambda}}{p}\\
	       &= \frac{t^{k-1}e^{-\frac{t}{\theta}}e^{-t\lambda}}{\theta^k\Gamma(k)\left(\frac{1}{1+\lambda\theta}\right)^k}\\
	       &= \frac{t^{k-1}e^{-t\frac{1+\lambda\theta}{\theta}}}{\left(\frac{\theta}{1+\lambda\theta}\right)^k\Gamma(k)}
 \end{align*}
 where $(a)$ is obtained by applying Bayes rule and in $(b)$, $p$ is given by \eqref{eq10}.  
\end{proof}

\begin{figure}[!t]
 \centering
 \includegraphics[scale=0.3]{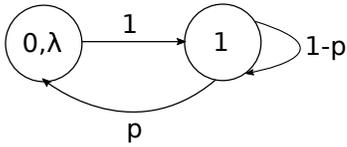}
 \caption{Semi-Markov chain representing the queue for LCFS with preemption}
 \label{fig2}
\end{figure}

In order to apply Lemma~\ref{lemma1}, we first notice that for a given packet $i$, the event $\{S_i<X_i\}$ is equivalent to the event \{packet $i$ was successfully received\}. Hence the probability $P=\mathbb{P}(S_i<\alpha|S_i<X_i)$ is the probability that the service time of the $i^{th}$ packet is less than $\alpha$ given that this packet was successfully transmitted. However, since the service times and interarrival times are i.i.d then $P$ does not depend on the index $i$. Now since $T$ is the service time of a successful packet then this leads us to
\begin{equation}
 \label{eq15}
 \mathbb{P}(T<\alpha) = \mathbb{P}(S_i<\alpha|S_i<X_i) = \mathbb{P}(S<\alpha|S<X),
\end{equation}
where $S$ and $X$ are the generic service and interarrival time respectively. By replacing $G$ by $S$ and $F$ by $X$ in Lemma~\ref{lemma1}, we deduce that the system time $T$ is gamma distributed with parameters $\left(k,\frac{\theta}{1+\lambda\theta}\right)$. Therefore,
\begin{equation}
 \label{eq16}
 \mathbb{E}(T) = \frac{k\theta}{1+\lambda\theta}.
\end{equation}
Now we turn our attention to the distribution of $Y$ for which we compute its moment generating function. Before going further in our analysis, we state the following lemma.
\begin{lemma}
\label{lemma2}
Let $G$ be gamma distributed with parameters ($k$,$\theta$) and $F$ be a rate $\lambda$ exponential random variable independent of $G$. If $F'$ is a random variable such that $$\mathbb{P}(F'<\alpha)=\mathbb{P}(F<\alpha|F<G),$$ then the moment generating function of $F'$ is given by
\begin{equation}
  \label{eq18}
  \phi_{F'}(s) = \frac{1}{1-p}\left(\frac{\lambda}{\lambda-s}-\frac{\lambda}{\lambda-s}\frac{1}{(1+\theta(\lambda-s))^k}\right),
\end{equation}
where $p=\left(\frac{1}{1+\lambda\theta}\right)^k$.
\end{lemma}
\begin{proof}
  We first start by computing the probability density function of $X'$.
  \begin{align*}
   f_{X'}(t) &= \lim_{\epsilon\to0}\frac{\mathbb{P}(t\leq F<t+\epsilon|F<G)}{\epsilon}\\
	     &= \lim_{\epsilon\to0}\frac{\mathbb{P}(t\leq F<t+\epsilon)\mathbb{P}(F<G|t\leq F<t+\epsilon)}{\epsilon\mathbb{P}(F<G)}\\
	     &= \frac{\lambda e^{-t\lambda}\mathbb{P}(G>t)}{1-p},
  \end{align*}
  where $p=\left(\frac{1}{1+\lambda\theta}\right)^k$.
  
  So now we can calculate the moment generating function of $F'$.
  \begin{align*}
   \phi_{F'}(s) &= \int_0^\infty f_{F'}(t)e^{st}\mathrm{d}t\\
		&= \int_0^\infty \frac{1}{1-p}\lambda e^{-t\lambda}\mathbb{P}(G>t)e^{st}\mathrm{d}t\\
		&= \frac{1}{1-p}\left(\frac{\lambda}{\lambda-s}-\lambda\int_0^\infty \mathbb{P}(G<t)e^{st}\mathrm{d}t\right)
  \end{align*}
  Using integration by parts and the fact that $\frac{\mathrm{d}}{\mathrm{dt}}\mathbb{P}(G<t)=f_G(t)=\frac{t^{k-1}e^{-\frac{t}{\theta}}}{\theta^k\Gamma(k)}$, we get
  \begin{equation*}
   \phi_{F'}(s) = \frac{1}{1-p}\left(\frac{\lambda}{\lambda-s} - \frac{\lambda}{(\lambda-s)(1+\theta(\lambda-s))^k}\right)
  \end{equation*}

\end{proof}

\begin{lemma}
The moment generating function of $Y$ is given by:
\begin{align}
 \label{eq20}
 \phi_Y(s) 
	   &= \frac{\lambda}{\lambda-s\left(1+\theta(\lambda-s)\right)^k}.
\end{align}
\end{lemma}
\begin{proof}
  By observing Figure~\ref{fig2} we notice that $Y$ is the smallest time needed to go from the 0-state back to the 0-state. Hence $Y$ can be written as $Y=X+W$ where $X$ is the generic interarrival time and $W$ is the time spent in the 1-state before the first jump back to the 0-state. So $W$ can be written as:
  \begin{align}
  \label{eq17}
  W &= \left\{\begin{array}{ll}
		S' & \text{with probability\ } p\\
		X'_1+S' & \text{with probability\ } (1-p)p\\
		X'_1+X'_2+S' & \text{with probability\ } (1-p)^2p\\
		\vdots
	      \end{array}\right.\nonumber\\
      &= \sum_{j=0}^M X'_j+S',
  \end{align}
  where $X'_j$ is such that $\mathbb{P}(X'_j<\alpha)=\mathbb{P}(X<\alpha|X<S)$, $S'$ is such that $\mathbb{P}(S'<\alpha)=\mathbb{P}(S<\alpha|S<X)$ and $M$ is a geometric$(p)$ random variable which is independent of $X'_j$ and $S'$, and which gives the number of discarded packets before the first successful reception. 
  Applying Lemmas~\ref{lemma1} and ~\ref{lemma2} on $S'$ and $X'$ respectively and using the fact that $M$,$S'$ and $X'_j$ are all mutually independent, it follows that
  \begin{align}
  \label{eq19}
  \phi_W(s) 
	    &= \mathbb{E}\left(e^{s\sum_{j=0}^M X'_j}\right)\phi_{S'}(s)\nonumber\\
	    &= \mathbb{E}\left(\phi_{X'}(s)^M\right)\left(\frac{1+\lambda\theta}{1+\theta(\lambda-s)}\right)^k\nonumber\\
	    &= \sum_{j=0}^\infty \phi_{X'}(s)^jp(1-p)^j\left(\frac{1+\lambda\theta}{1+\theta(\lambda-s)}\right)^k\nonumber\\
	    &= \frac{\lambda-s}{\lambda-s\left(1+\theta(\lambda-s)\right)^k}.
  \end{align}
  Moreover, since $X$ and $W$ are independent and $\phi_X(s)=\frac{\lambda}{\lambda-s}$, we get using \eqref{eq19}
  \begin{align*}
  \phi_Y(s) 
	    &= \phi_X(s)\phi_W(s)= \frac{\lambda}{\lambda-s\left(1+\theta(\lambda-s)\right)^k}.
  \end{align*}
\end{proof}

Now that we have found $\phi_Y$ we can compute the first two moments of $Y$ as $\mathbb{E}(Y) = \frac{(1+\lambda\theta)^k}{\lambda}$ and $\mathbb{E}(Y^2) = \frac{2(1+\lambda\theta)^{k-1}}{\lambda^2}\left((1+\lambda\theta)^{k+1}-k\theta\lambda\right)$.
Combining these results with \eqref{eq16}, we obtain,
\begin{equation}
 \label{eq23}
 \mathbb{E}(Q_i) = \frac{(1+\lambda\theta)^{2k}}{\lambda^2}.
\end{equation}

\subsubsection{Computing the effective rate}
Using \eqref{eq4} we get
\begin{align}
 \label{eq24}
 \lambda_e 
           &= \lambda p = \lambda\left(\frac{1}{1+\lambda\theta}\right)^k.
\end{align}

Now we are ready to compute the average age:
We conclude 
\begin{proposition}
The average age in the LCFS with preemption scheme assuming $\Gamma(k,\theta)$ service time is given by:
\begin{equation}
 \label{eq25}
 \Delta = \lambda_e\mathbb{E}(Q_i) = \frac{(1+\lambda\theta)^k}{\lambda}.
\end{equation}
\end{proposition}
\begin{proof}		
Using \eqref{eq23} and ~\eqref{eq24}.
\end{proof}

\subsection{Average peak age}
\begin{proposition}
The average peak age in the LCFS with preemption scheme assuming $\Gamma(k,\theta)$ service time is given by:
\begin{align}
 \label{eq26}
 \mathbb{E}(P_i) = \mathbb{E}(T) + \mathbb{E}(Y)= \frac{k\theta}{1+\lambda\theta}+\frac{(1+\lambda\theta)^k}{\lambda}.
\end{align}
\end{proposition}
\begin{proof}
Using \eqref{eq9}, ~\eqref{eq16} and the value of $\mathbb{E}(Y)$.
\end{proof}

\section{Age of information for LCFS without preemption}
\label{nopreempt}
Another interesting scheme worth to study is the LCFS without preemption. In this scenario, we assume that the queue has a buffer of size 1 and we wait for the packet being served to finish before serving a new one. If while serving a packet a new update arrives, it replaces any packet waiting in the buffer. In this section we will derive a closed form expression for the average age $\Delta$  and the average peak age $\mathbb{E}(P_k)$ for LCFS without preemption while assuming an Erlang distribution for the service time with parameter $(k,\theta)$. An Erlang distribution is nothing but a special case of the gamma distribution where $k\in\mathbb{N}$. Moreover, an Erlang distribution $(k,\theta)$ can be seen as the sum of $k$ independent memoryless random variables $A_j$, each with rate $\frac{1}{\theta}$. Using this observation, we model the state of the queue as a two-level Markov chain as shown in Figure~\ref{fig3}.
\begin{figure}[!t]
 \centering
 \includegraphics[scale=0.15]{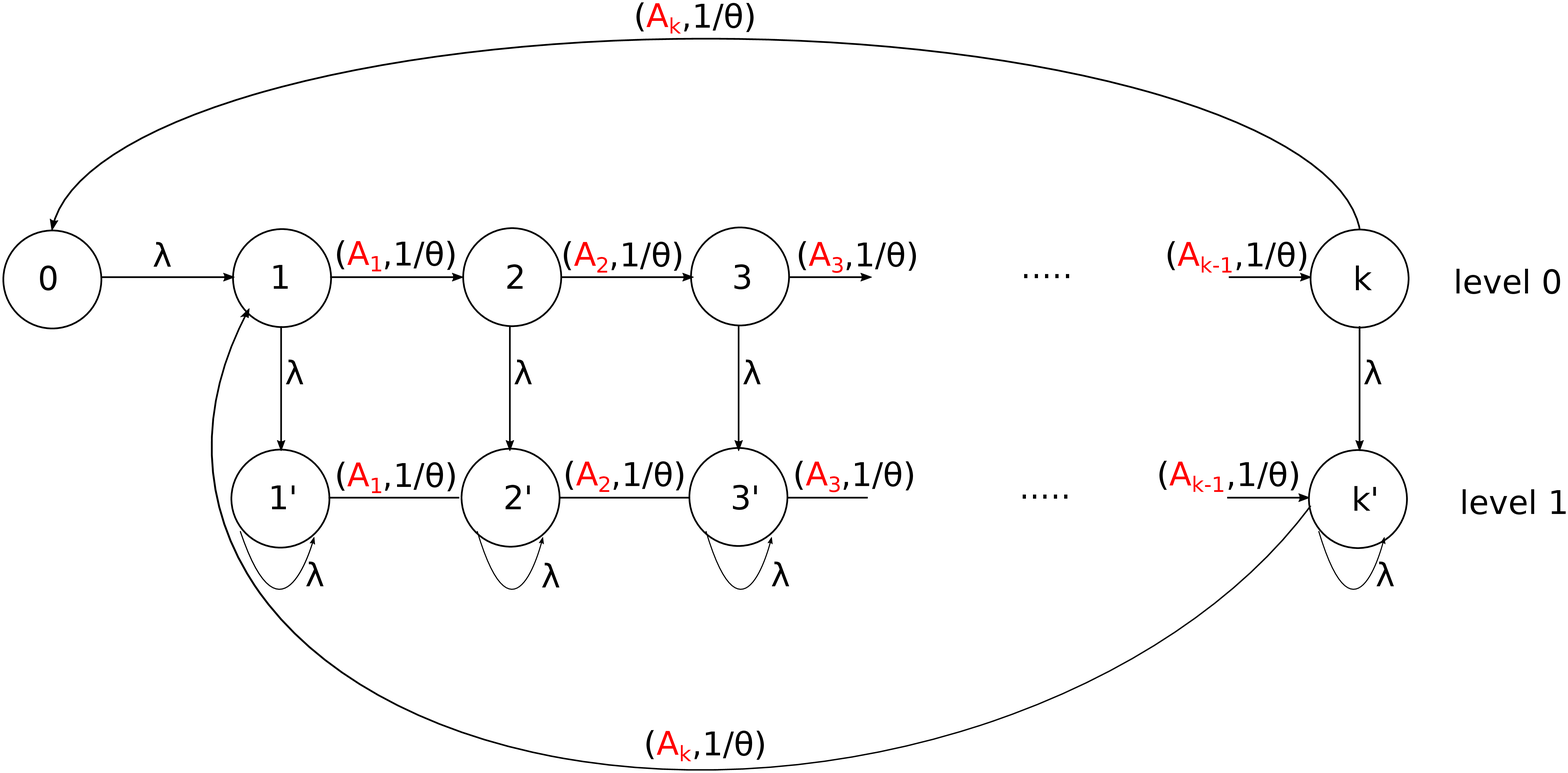}
 \caption{Markov chain representing the queue for LCFS without preemption}
 \label{fig3}
\end{figure}

As in the previous section, we will denote the generic rate-$\lambda$ interarrival time by $X$ and the generic Erlang distributed service time by $S=\sum_{j=1}^k A_j$. Using this notation, we notice that the service time can be represented as the succession of $k$ exponential-time steps that need to be accomplished for a successful reception. Hence, a packet in state $j\in \{1,\ldots,k\}$ or $j'\in \{1',\ldots,k'\}$ is a packet completing his $j^{th}$ step out of a total of $k$. Moreover, the 0-state represents an empty queue, all the states of level 0 represent an empty buffer and those of level 1 represent a full buffer. After spending an exponential amount of time in the 0-state, we can only jump to the 1-state once a new update arrives. Using the memoryless property of the exponential distribution, we can describe the evolution of this packet in the queue as follows: at state $j\in\{1,\ldots,k\}$, two exponential clocks start simultaneously. One clock -- denoted $A_j$ -- of rate $\frac{1}{\theta}$ and another one -- denoted $\Lambda_j$ -- of rate $\lambda$. If clock $A_j$ ticks first then the packet jump to state $j+1$ and the buffer stays empty. Otherwise it jumps to state $j'$ since now the buffer is full. On the other hand, if the packet is at state $j'$ and the $A_j$ clock ticks first then the packet jump to state $(j+1)'$ without updating the buffer. However, if the $\Lambda_j$ ticks first then the packet stays in state $j'$ but we update the buffer with the new arrival.

\subsection{Average age}
\begin{proposition}
The average age in the LCFS without preemption scheme assuming Erlang $E(k,\theta)$ service time is:
\begin{align}
 \label{eq37}
 \Delta &= \frac{k\theta(2+\lambda\theta+3k\lambda\theta)}{2(q^k+k\lambda\theta)} + \frac{2(1-k^2\lambda\theta)}{\lambda(1+k\lambda\theta(1+\lambda\theta)^k)}\nonumber\\
        & +\frac{k\theta(1+k\lambda\theta+2k)}{1+\lambda\theta+k\lambda\theta(1+\lambda\theta)^{k+1}}\nonumber\\
        &- \frac{1+\lambda\theta+k\lambda\theta}{\lambda(1+\lambda\theta)\left((1+\lambda\theta)^k+k\lambda\theta(1+\lambda\theta)^{2k}\right)}
\end{align}
\end{proposition}
\begin{proof}
  As in the previous section we need to compute the effective rate (given by \eqref{eq4}) and $\mathbb{E}(Q_i)$ (given by \eqref{eq6}).
  \subsubsection{Computing $\mathbb{E}(Q_i)$}
  Following the same line of thoughts as in Section~\ref{preempt}, we will calculate $\mathbb{E}(T_{i-1}Y_{i-1})$ by expressing it as the average of two conditionally independent variables given some set of events. For this end we define the family of events $\Psi_j^{i}=\left\{A_j^{i}>\Lambda_j^{i}; \sum_{l=j+1}^k A^{i}_l<X\right\}$, where $1\leq j \leq k$. Hence $\Psi_j^{i}$ is the event that during the service time of the $i^{th}$ successful packet a new update arrived at the $j^{th}$ step of the service time (i.e, state $j$ or $j'$) and then no new update arrived for the remainder of the service time. The superscript $(i)$ is used to indicate that we are dealing with the $i^{th}$ successful packet. For $j=0$, $\Psi_0^i$ is the event that the $i^{th}$ successful packet leaves the queue empty. Note that for every $i$, $\{\Psi_j^{i},1\leq j \leq k\}$ is a partition of the probability space.
  
  It is sufficient to condition on the event $\Psi_0^{i-1}$ in order to ensure conditional independence between $T_{i-1}$ and $Y_{i-1}$. This is due to the following fact: given $\Psi_0^{i-1}$, we know that the $(i-1)^{th}$ successful packet left the queue empty and hence we have a situation identical to that of the with preemption case (see Section~\ref{preempt}) and $T_{i-1}$ and $Y_{i-1}$ are independent. On the other hand, given $\overline{\Psi_0^{i-1}}$, the buffer is not empty and thus a new packet will be served directly after the departure of the $(i-1)^{th}$ successful packet. In this case, the interdeparture time $Y_{i-1}$ is simply the service time of the $i^{th}$ successful packet whose value is independent of $T_{i-1}=W_{i-1}+S_{i-1}$, where $W_{i-1}$ the waiting time and $S_{i-1}$ is the service time of the $(i-1)^{th}$ successful packet (see Figure~\ref{fig1b}).
  
  Although conditioning on $\Psi_0^{i-1}$ is enough to obtain independence between $T_{i-1}$ and $Y_{i-1}$, we will need to condition on the two independent events $\Psi_j^{i-1}$ and $\Psi_l^{i-2}$ in order to be able to calculate the conditional expectation of $T_{i-1}$. However, it is clear that conditioning on these two events also leads to the independence between $T_{i-1}$ and $Y_{i-1}$. Hence we get
  \begin{align}
  \label{eq27}
  \mathbb{E}&(T_{i-1}Y_{i-1}) \nonumber\\
	    &= \sum_{j,l=0}^k\left(\mathbb{E}\left(T_{i-1}|\Psi_j^{i-1}\Psi_l^{i-2}\right)\mathbb{E}\left(Y_{i-1}|\Psi_j^{i-1}\Psi_l^{i-2}\right)\right.\nonumber\\
	    &\qquad\qquad \left.\vphantom{\mathbb{E}\left(T_{i-1}|\Psi_j^{i-1}\Psi_l^{i-2}\right)}\times\mathbb{P}(\Psi_j^{i-1})\mathbb{P}(\Psi_l^{i-2})\right) .
  \end{align} 
  We start by computing $\mathbb{E}\left(T_{i-1}|\Psi_j^{i-1}\Psi_l^{i-2}\right)=\mathbb{E}\left(W_{i-1}|\Psi_j^{i-1}\Psi_l^{i-2}\right)+\mathbb{E}\left(S_{i-1}|\Psi_j^{i-1}\Psi_l^{i-2}\right)$.
  
  The waiting time of the $(i-1)^{th}$ successful packet doesn't depend on $\Psi_j^{i-1}$ since they are disjoint in time, but it does depend on $\Psi_l^{i-2}$. In fact, given $\Psi_0^{i-2}$, the $(i-1)^{th}$ successful packet will not wait and start service upon arrival since the $(i-2)^{th}$ successful packet left the queue empty. However, given $\Psi_l^{i-2}$ with $l\neq 0$, the $(i-1)^{th}$ successful packet arrived when the $(i-2)^{th}$ successful packet was at state $l$ or $l'$ of its service time. In order to find the distribution of $W_{i-1}$ conditioned on $\Psi_l^{i-2}$ we introduce the following event: $\Psi_{l,n}^{i} = \left\{\sum_{g=1}^n \Lambda_{l,g}^{i}<A_l^{i}, \sum_{g=1}^{n+1}\Lambda_{l,g}^{i}>\sum_{m=l}^k A_m^{i}\right\}$, where $\{\Lambda_{l,g}^{i}\}_{g\geq1}$ is the sequence of interarrival times after the $(i)^{th}$ successful packet enters state $l$. Notice that $\Psi_{l,n}^{i}$ is the event that exactly $n$ updates arrived when the $i^{th}$ successful packet was in state $l$ (or $l'$) and then no more updates were generated for the remainder of the service time. Hence $\Psi_l^{i}=\cup_{n=1}^\infty \Psi_{l,n}^{i}$. So conditioned on $\Psi_{l,n}^{(i-2)}$ we have 
  \begin{align}
  \label{eq28a}
  W_{i-1} &= \sum_{m=l}^k A_m^{(i-2)}-\sum_{g=1}^n \Lambda_{l,g}^{(i-2)}\nonumber\\
	  &= (A_l^{(i-2)}-\sum_{g=1}^n \Lambda_{l,g}^{(i-2)}) + \sum_{m=l+1}^k A_m^{(i-2)}
  \end{align}
  It can be shown that, conditioned on $\{\sum_{g=1}^n \Lambda_{l,g}^{i}<A_l^{i}\}$, $(A_l^{(i-2)}-\sum_{g=1}^n \Lambda_{l,g}^{(i-2)})$ has an exponential distribution with rate $\frac{1}{\theta}$. This means that under this condition alone, $W_{i-1}$ has the same distribution as the sum of $k-l+1$ independent exponential random variables with rate $\frac{1}{\theta}$. If we further condition on $\left\{\sum_{g=1}^{n+1}\Lambda_{l,g}^{i}>\sum_{m=l}^k A_m^{i}\right\}$ and use Lemma~\ref{lemma1}, we deduce that conditioned on $\Psi_{l,n}^{(i-2)}$, $W_{i-1}$ has a gamma distribution with parameters $\left(k-l+1, \frac{\theta}{1+\lambda\theta}\right)$. Now since $\Psi_l^{i-2}=\cup_{n=1}^\infty \Psi_{l,n}^{i-2}$, we conclude that if we condition on $\Psi_l^{i-2}$, $W_{i-1}$ is distributed as $\Gamma\left(k-l+1, \frac{\theta}{1+\lambda\theta}\right)$. Therefore,
  \begin{equation}
  \label{eq28}
  \mathbb{E}\left(W_{i-1}|\Psi_j^{i-1}\Psi_l^{i-2}\right) = \left\{ \begin{array}{ll}
								      0 & \text{if\ } l=0\\
								      \frac{(k-l+1)\theta}{1+\lambda\theta} & \text{if\ } l\neq 0
								    \end{array}\right. .
  \end{equation}
  Now we turn our attention to $\mathbb{E}\left(S_{i-1}|\Psi_j^{i-1}\Psi_l^{i-2}\right)$. One first notices that the service time $S_{i-1}$ of the $(i-1)^{th}$ successful packet is independent of its arrival time given by the event $\Psi_l^{i-2}$ since we assumed independence between service time and interarrival time. Hence, $\mathbb{E}\left(S_{i-1}|\Psi_j^{i-1}\Psi_l^{i-2}\right)=\mathbb{E}\left(S_{i-1}|\Psi_j^{i-1}\right)$. For the case $j=0$, we get 
  \begin{align}
  \label{eq29}
  \mathbb{E}\left(S_{i-1}|\Psi_0^{i-1}\right)&= \mathbb{E}\left(\sum_{m=1}^kA_m^{i-1}|\sum_{m=1}^kA_m^{i-1}<X\right)\nonumber\\
					      &= \frac{k\theta}{1+\lambda\theta}
  \end{align}
  where the last equality is obtained by applying Lemma~\ref{lemma1} with $G=\sum_{m=1}^kA_m^{i-1}$ and $F=X$. As for the case $j\neq0$, we get
  \begin{align}
  \label{eq30}
  \mathbb{E}&\left(S_{i-1}|\Psi_j^{i-1}\right)\nonumber\\
  &= \mathbb{E}\left(\sum_{m=1}^kA_m^{i-1}|A_j^{i-1}>\Lambda_j^{i-1}, \sum_{m=j+1}^k A^{i-1}_m<X\right)\nonumber\\
  &= \sum_{m=1}^{j-1}\mathbb{E}(A_m^{i-1})+\mathbb{E}(A_j^{i-1}|A_j^{i-1}>\Lambda_j^{i-1})\nonumber\\
  & \qquad +\mathbb{E}\left(\sum_{m=j+1}^kA^{i-1}_m\Bigg|\sum_{m=j+1}^kA^{i-1}_m<X\right)\nonumber\\
  &\overset{(a)}{=} (j-1)\theta + \frac{\theta(2+\lambda\theta)}{1+\lambda\theta}+\frac{(k-j)\theta}{1+\lambda\theta}\nonumber\\
  &= \frac{\theta(1+k+j\lambda\theta)}{1+\lambda\theta}
  \end{align}
  where the third term in $(a)$ is obtained by applying Lemma~\ref{lemma1} with $G=\sum_{m=j+1}^kA^{i-1}_m$ and $F=X$. Therefore, combining \eqref{eq29} and ~\eqref{eq30} we get,
  \begin{equation}
  \label{eq31}
  \mathbb{E}\left(T_{i-1}|\Psi_j^{i-1}\Psi_l^{i-2}\right) = \left\{\begin{array}{ll}
								    \frac{k\theta}{1+\lambda\theta} & \text{if\ } l=0, j=0\\
								    \frac{\theta(k+1+j\lambda\theta)}{1+\lambda\theta} & \text{if\ } l=0, j>0\\
								    \frac{\theta(2k-l+1)}{1+\lambda\theta} & \text{if\ } l>0, j=0\\
								    \frac{\theta(2k-l+2+j\lambda\theta)}{1+\lambda\theta} & \text{if\ } l>0, j>0\\
								    \end{array}
  \right. .
  \end{equation}
  Now we need to compute $\mathbb{E}\left(Y_{i-1}|\Psi_j^{i-1}\Psi_l^{i-2}\right)$. For this end, observe that $Y_{i-1}$ is independent of $\Psi_l^{i-2}$ given that they don't overlap in time. Moreover, for $j=0$, the $(i-1)^{th}$ successful packet leaves the queue empty and thus we will need to wait an exponential amount of time $X'$ of rate $\lambda$ before the $i^{th}$ successful packet arrives and is served directly. Hence, conditioned on $\Psi_0^{i-1}$, $Y_{i-1}$ has same distribution as $(X'+S)$ with $X'$ and $S$ independent. On the other hand, for $j\neq 0$, the $(i-1)^{th}$ successful packet leaves the queue with another packet waiting in the buffer ready to be served. Thus in this case, $Y_{i-1}$ is simply the service time of the $i^{th}$ successful packet. To sum up,
  \begin{equation}
  \label{eq32}
  \mathbb{E}\left(Y_{i-1}|\Psi_j^{i-1}\Psi_l^{i-2}\right) = \left\{\begin{array}{ll}
								    \frac{1}{\lambda}+k\theta & \text{if\ }j=0\\
								    k\theta & \text{if\ }j>0
								    \end{array}
  \right.
  \end{equation}
  To compute $\mathbb{E}(T_{i-1}Y_{i-1})$ we still need the probability $\mathbb{P}(\Psi_j^{i-1})$. For $j>0$, we use the fact that $\Psi_j^{i-1}$ is the intersection of two independent events and find that $\mathbb{P}(\Psi_j^{i-1})=\frac{\lambda\theta}{(1+\lambda\theta)^{k-j+1}}$. As for $j=0$, we have already seen in Section~\ref{preempt} that $\mathbb{P}(\Psi_0^{i-1})=p=\left(\frac{1}{1+\lambda\theta}\right)^k$. These probabilities are independent of the index $i$ and thus we can find $\mathbb{P}(\Psi_l^{i-2})$ by replacing $j$ by $l$ in the previous expressions. Combining this results with \eqref{eq31}, ~\eqref{eq32} we obtain after some tedious calculations
  \begin{align}
  \label{eq33}
  \mathbb{E}&(T_{i-1}Y_{i-1}) = \frac{k\theta}{\lambda}(1+k\lambda\theta)+q^k\left(\frac{1-k\lambda\theta(2k+1)}{\lambda^2}\right)\nonumber\\
	    &+q^{k+1}\left(\frac{k\theta(1+k\lambda\theta+2k)}{\lambda}\right)-\frac{1}{\lambda^2}q^{2k}-\frac{k\theta}{\lambda}q^{2k+1}
  \end{align}
  with $q=\frac{1}{1+\lambda\theta}$.\\
  The last term to compute in order to obtain $\mathbb{E}(Q_i)$ is $$\mathbb{E}(Y_{i-1}^2)= \mathbb{E}(Y_{i-1}^2|\Psi_0^{i-1})\mathbb{P}(\Psi_0^{i-1})+\mathbb{E}(Y_{i-1}^2|\overline{\Psi_0^{i-1}})\mathbb{P}(\overline{\Psi_0^{i-1}}).$$ Based on our previous observations we know that $\mathbb{E}(Y_{i-1}^2|\Psi_0^{i-1})=\mathbb{E}((X'+S)^2)$ and $\mathbb{E}(Y_{i-1}^2|\overline{\Psi_0^{i-1}})=\mathbb{E}(S^2)$. Using these facts we get
  \begin{equation}
  \label{eq34}
  \mathbb{E}(Y_{i-1}^2) = k\theta^2+k^2\theta^2+q^k\left(\frac{2+2k\lambda\theta}{\lambda^2}\right).
  \end{equation}
  Combining \eqref{eq33} and ~\eqref{eq34}, we finally get
  \begin{align}
  \label{eq35}
  \mathbb{E}&(Q_i) = \frac{k\theta(2+\lambda\theta+3k\lambda\theta)}{2\lambda}+2q^k\left(\frac{1-k^2\lambda\theta}{\lambda^2}\right)\nonumber\\
	    & +q^{k+1}\left(\frac{k\theta(1+k\lambda\theta+2k)}{\lambda}\right)-\frac{1}{\lambda^2}q^{2k}-\frac{k\theta}{\lambda}q^{2k+1}.
  \end{align}

  \subsubsection{Computing the effective rate}
  To calculate the effective rate we first observe that the event \{packet is successfully received\} is equivalent to the event \{packet passes by the 1-state\}. Hence if we \lq uniformize\rq\ the Markov chain so that the time spent at each state is exponential with rate $\lambda+\frac{1}{\theta}$, we get $\lambda_e = \left(\lambda+\frac{1}{\theta}\right)\pi_1$ where $\pi_1$ is the steady-state probability of the 1-state in the \lq uniformized\rq\ Markov chain. The analysis of such chain (\cite{ross}, chapter 5) gives $\pi_1=\frac{q(1-q)}{q^{k+1}+k(1-q)}$. Therefore,
  \begin{equation}
  \label{eq36}
  \lambda_e = \frac{\lambda(1+\lambda\theta)^k}{1+k\lambda\theta(1+\lambda\theta)^k}.
  \end{equation}
  
  Finally, replacing  $\mathbb{E}(Q_i)$ and $\lambda_e$ in $\Delta=\lambda_e\mathbb{E}(Q_i)$ by their expressions in \eqref{eq35} and ~\eqref{eq36}, we obtain our result.
\end{proof}

\subsection{Average peak  age}
\begin{proposition}
The average peak age in the LCFS without preemption scheme assuming Erlang $E(k,\theta)$ service time is:
\begin{equation}
 \label{eq40}
 \mathbb{E}(P_i) = \frac{1}{\lambda}+2k\theta-\frac{k\theta}{(1+\lambda\theta)^{k+1}}.
\end{equation}
\end{proposition}
\begin{proof}
  We know that $\mathbb{E}(P_i)=\mathbb{E}(T_{i-1})+\mathbb{E}(Y_{i-1})$. We calculate these two terms as follows
  \begin{align}
  \label{eq38}
  \mathbb{E}(T_{i-1})&= \sum_{j,l=0}^k \mathbb{E}\left(T_{i-1}|\Psi_j^{i-1}\Psi_l^{i-2}\right)\mathbb{P}(\Psi_j^{i-1})\mathbb{P}(\Psi_l^{i-2})\nonumber\\
		      &= \frac{1}{\lambda}+k\theta-q^{k+1}\left(\frac{1+\lambda\theta+k\lambda\theta}{\lambda}\right),
  \end{align}
  where we used \eqref{eq31} for the last equality. For $\mathbb{E}(Y_{i-1})$ we will only condition on $\Psi_0^{i-1}$. Hence using \eqref{eq32}, we get
  \begin{align}
  \label{eq39}
  \mathbb{E}(Y_{i-1}) &= \mathbb{E}(Y_{i-1}|\Psi_0^{i-1})\mathbb{P}(\Psi_0^{i-1})+\mathbb{E}(Y_{i-1}|\overline{\Psi_0^{i-1}})\mathbb{P}(\overline{\Psi_0^{i-1}})\nonumber\\
		      &= k\theta + \frac{q^k}{\lambda}.
  \end{align}
  Thus, combining the above two results we obtain our result.
\end{proof}

\section{Age of information for deterministic service time}
\label{deterministic}
In order to compute the four ages of interest under a deterministic service time assumption, we use Lemma~\ref{lemma3}. For that, we fix the mean of the service times $S_n$ to $\mathbb{E}(S_n)= \frac{1}{\mu}$, for some $\mu>0$, and let $k\to\infty$. It is beyond the scope of this paper to show that if $S_n\overset{d}{\to} Z$, as $k\to\infty$ then we also have convergence in the average ages, i.e, $\Delta_{S_n}\to\Delta_Z$. Here $\Delta_{S_n}$ refers to the average age corresponding to service time $S_n$. However, we will use this result to derive the different ages.

\subsection{LCFS with preemption}
Letting $k\to\infty$ in \eqref{eq25} and ~\eqref{eq26}, we get
\begin{align}
 \label{eq41}
 \Delta &= \frac{e^{\lambda/\mu}}{\lambda}\\
 \mathbb{E}(P_i) &= \frac{1}{\mu}+\frac{e^{\lambda/\mu}}{\lambda}
\end{align}

\subsection{LCFS without preemption}
Letting $k\to\infty$ in \eqref{eq37} and ~\eqref{eq40}, we get
\begin{align}
 \label{eq42}
 &\Delta = \frac{2(2+\rho-\rho^2)-2e^{-\rho}(1+\rho)+\rho e^{\rho}(2+3\rho)}{2\lambda\left(1+\rho e^{\rho}\right)}\\
 &\mathbb{E}(P_i) = \frac{1}{\lambda}+\frac{2-e^{-\rho}}{\mu}
\end{align}
where $\rho=\frac{\lambda}{\mu}$.

\section{Numerical results}
\label{numerical}
In this section we show that the theoretical results obtained in the previous sections match the simulations. We also compare the performance of the two transmission schemes of interest as well as the effect of the parameter $k$ on each of them. First it is worthy to specify that all simulations were done using gamma distributed service times with all having the same mean $k\theta=1$, except for the deterministic case where the service time is fixed to 1. Figure~\ref{fig4} presents the average age under LCFS with preemption scheme and gamma distributed service time. Two observations can be made based on this plot: $(i)$ the theoretical curves given by \eqref{eq25} and ~\eqref{eq41} coincide with the empirical curves and $(ii)$ as the value of $k$ increases, the average age increases for all values of $\lambda$. This means that, under LCFS with preemption, the average age assuming deterministic service time $(k\to\infty)$ is higher than the average age assuming a regular gamma distributed service. In particular, it is higher than the average age assuming memoryless time. This observation can be explained by the fact that the probability of packet being preempted is given by $1-p=1-\left(\frac{1}{1+\lambda\theta}\right)^k$ (refer to Section~\ref{preempt}) which is an increasing function of $k$. Therefore, as $k$ increases the receiver will have to wait on average a longer time till a new update is delivered since the preempting rate becomes higher. This analysis is true for any value of $\lambda$, hence the phenomenon seen in Figure~\ref{fig4}.

In a parallel setting, Figure~\ref{fig5} presents the average age under LCFS without preemption. In this case also two observations can be made: $(i)$ the theoretical curves given by \eqref{eq37} and ~\eqref{eq42} match the empirical results and $(ii)$ as the value of $k$ increases, the average age decreases for almost all $\lambda$ (except for values close to 0 where all distributions behave similarly). This difference in performance is especially seen at high $\lambda$. We give here a quick intuition that explains this behavior. When $\lambda$ is high $(\lambda\to\infty)$, the time where the queue is empty goes to 0 and thus the queue is always transmitting. This also means that on average the waiting time $W_{i-1}$ goes to 0. Given these two observations, one can say that the system time $T_{i-1}$ and the interdeparture time $Y_{i-1}$ will have almost the same distribution as the service time, while being almost independent. Thus $\mathbb{E}(Q_i)\overset{\lambda\to\infty}{\longrightarrow}\mathbb{E}(S)^2+\frac{\mathbb{E}(S^2)}{2}$. As for the effective rate $\lambda_e$, since the queue is almost always busy, the average rate at which the receiver gets new update is nothing but the inverse of the average service time, i.e $\lambda_e\overset{\lambda\to\infty}{\longrightarrow}\frac{1}{\mathbb{E}(S)}$. Therefore, $\Delta\overset{\lambda\to\infty}{\longrightarrow}\mathbb{E}(S)+\frac{\mathbb{E}(S^2)}{2\mathbb{E}(S)}=\frac{\theta}{2}+\frac{3k\theta}{2}$. This result --- which is also obtained by taking the limit over $\lambda$ in \eqref{eq37} --- is decreasing with $k$. Hence the behavior seen in Figure~\ref{fig5}.

Next, we compare the performance of the two transmission schemes in two models: for gamma distributed and deterministic service time. Figure~\ref{fig6} shows the average age under LCFS with and without preemption when the service time is taken to be gamma distributed with $k=2$. In this case we notice that for small $\lambda$ the two schemes perform similarly. However, for $\lambda$'s around 1, the LCFS with preemption scheme performs slightly better before being outperformed by the LCFS without preemption scheme at high $\lambda$'s. Practically, this means that if one is using a medium whose service time is modeled as a gamma random variable, the best strategy (among the considered ones) is not to preempt while increasing the update generation rate as much as possible. This strategy also applies when the service time is deterministic as seen in Figure~\ref{fig7}. In fact, we observe that for deterministic service time and for all values of $\lambda$, the average age and the average peak age for the LCFS without preemption scheme are smaller than the average age and average peak age for the LCFS with preemption respectively.

\begin{figure}[!t]
 \centering
 \includegraphics[scale=0.47]{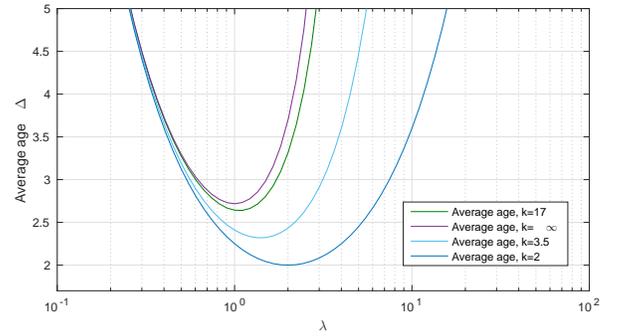}
 \caption{Average age for gamma service time $S$ with $\mathbb{E}(S)=1$, different $k$ and LCFS with preemption}
 \label{fig4}
\end{figure}

\begin{figure}[!t]
 \centering
 \includegraphics[scale=0.47]{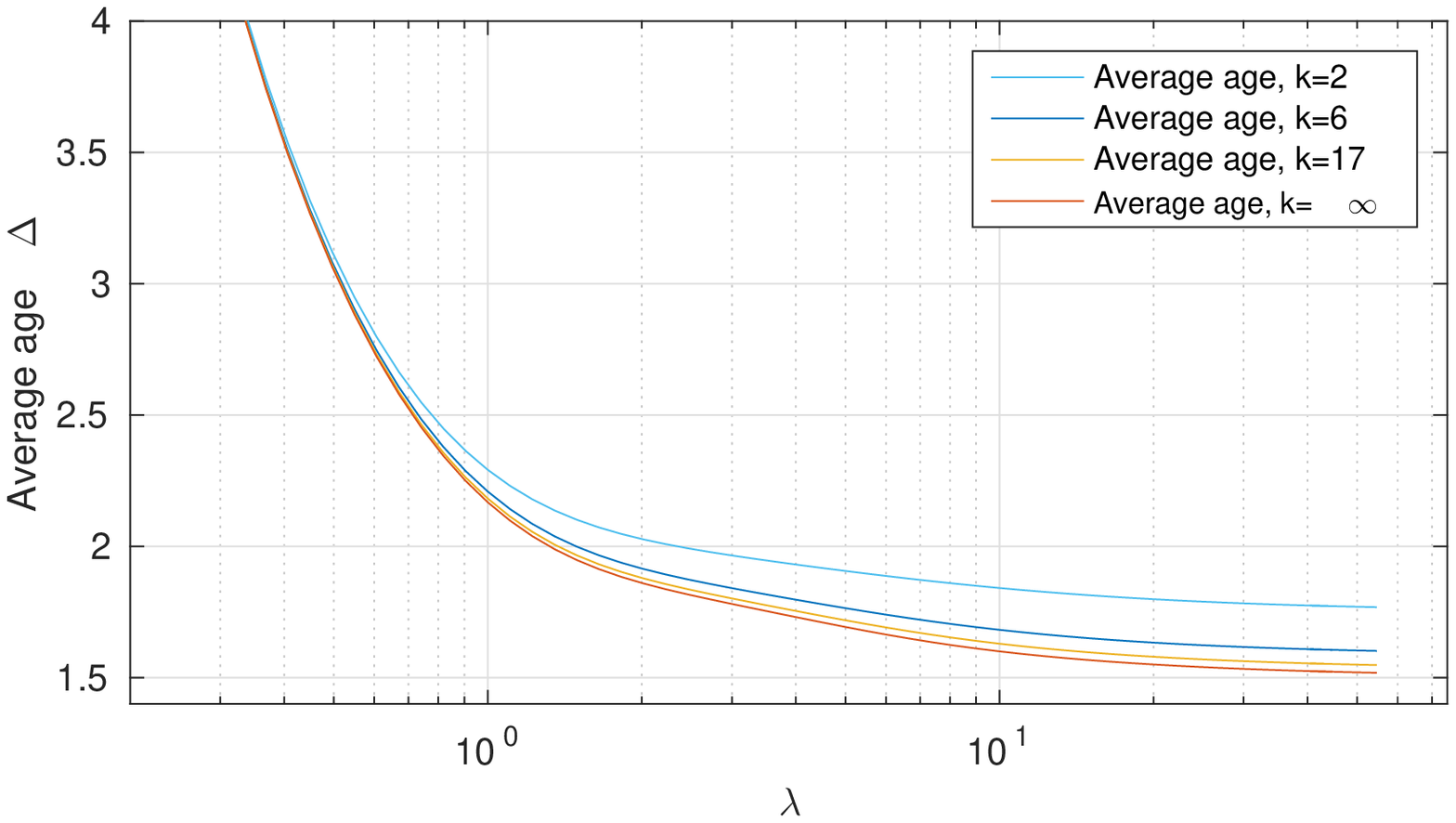}
 \caption{Average age for gamma service time $S$ with $\mathbb{E}(S)=1$, different $k$ and LCFS without preemption}
 \label{fig5}
\end{figure}

\begin{figure}[!t]
 \centering
 \includegraphics[scale=0.47]{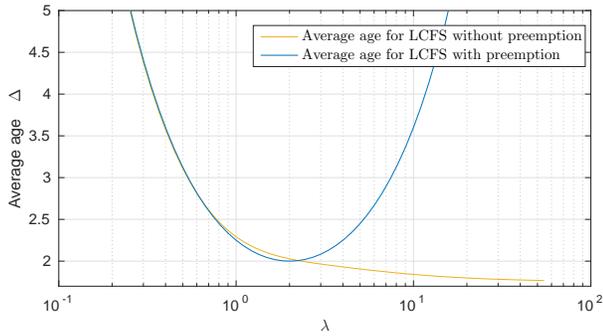}
 \caption{Average age for gamma service time $S$ with $k=2$ and $\mathbb{E}(S)=1$}
 \label{fig6}
\end{figure}

\begin{figure}[!t]
 \centering
 \includegraphics[scale=0.47]{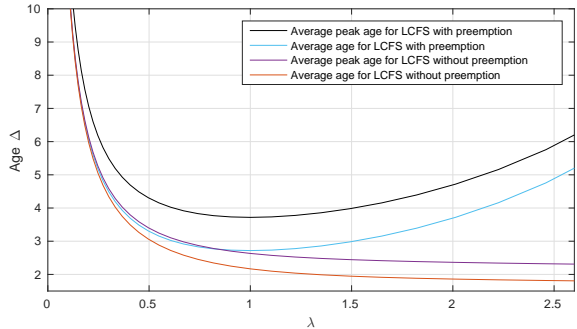}
 \caption{Average age and average peak age for deterministic service time}
 \label{fig7}
\end{figure}


\section{Conclusion}
\label{conclusion}

We considered the gamma distribution as a model for the service time in status update systems. We computed and analyzed the average and average peak age of information under two schemes: LCFS with preemption and LCFS without preemption. This allowed us to evaluate these metrics for deterministic service time. This suggests that considering gamma distributions for similar problems can be a good idea since the Gamma distributions (or at least Erlang distributions) are practically relevant as they can be used to model the total service time for relay networks.



\section*{Acknowledgment}

The authors would like to thank Emre Telatar for helpful discussions.



\bibliographystyle{IEEEtran}
\bibliography{IEEEabrv,Age_of_information_gamma_service_time_with_without_preemption}

\begin{thebibliography}{1}
\providecommand{\url}[1]{#1}
\csname url@samestyle\endcsname
\providecommand{\newblock}{\relax}
\providecommand{\bibinfo}[2]{#2}
\providecommand{\BIBentrySTDinterwordspacing}{\spaceskip=0pt\relax}
\providecommand{\BIBentryALTinterwordstretchfactor}{4}
\providecommand{\BIBentryALTinterwordspacing}{\spaceskip=\fontdimen2\font plus
\BIBentryALTinterwordstretchfactor\fontdimen3\font minus
  \fontdimen4\font\relax}
\providecommand{\BIBforeignlanguage}[2]{{%
\expandafter\ifx\csname l@#1\endcsname\relax
\typeout{** WARNING: IEEEtran.bst: No hyphenation pattern has been}%
\typeout{** loaded for the language `#1'. Using the pattern for}%
\typeout{** the default language instead.}%
\else
\language=\csname l@#1\endcsname
\fi
#2}}
\providecommand{\BIBdecl}{\relax}
\BIBdecl

\bibitem{realtime}
S.~Kaul, R.~Yates, and M.~Gruteser, ``Real-time status: How often should one
  update?'' in \emph{INFOCOM, 2012 Proceedings IEEE}, March 2012, pp.
  2731--2735.

\bibitem{multiple}
R.~Yates and S.~Kaul, ``Real-time status updating: Multiple sources,'' in
  \emph{Information Theory Proceedings (ISIT), 2012 IEEE International
  Symposium on}, July 2012, pp. 2666--2670.

\bibitem{random}
C.~Kam, S.~Kompella, and A.~Ephremides, ``Age of information under random
  updates,'' in \emph{Information Theory Proceedings (ISIT), 2013 IEEE
  International Symposium on}, July 2013, pp. 66--70.

\bibitem{queues}
S.~Kaul, R.~Yates, and M.~Gruteser, ``Status updates through queues,'' in
  \emph{Information Sciences and Systems (CISS), 2012 46th Annual Conference
  on}, March 2012, pp. 1--6.

\bibitem{ontheage}
\BIBentryALTinterwordspacing
M.~Costa, M.~Codreanu, and A.~Ephremides, ``On the age of information in status
  update systems with packet management,'' \emph{CoRR}, vol. abs/1506.08637,
  2015. [Online]. Available: \url{http://arxiv.org/abs/1506.08637}
\BIBentrySTDinterwordspacing

\bibitem{ross}
S.~M. Ross, \emph{{Stochastic Processes (Wiley Series in Probability and
  Statistics)}}, 2nd~ed.\hskip 1em plus 0.5em minus 0.4em\relax Wiley, Feb.
  1995.

\end{thebibliography}
%



\end{document}